\documentclass{article}
\usepackage[utf8]{inputenc}

\usepackage{fullpage}
\usepackage{microtype}
\usepackage{graphicx}
\usepackage{subfigure}
\usepackage{booktabs}
\usepackage{physics}
\usepackage{hyperref}
\usepackage{color}
\usepackage{bm}
\usepackage{authblk}
\usepackage[ruled,vlined]{algorithm2e}

\usepackage{amsmath,amssymb,amsfonts,amsthm}
\usepackage{mathtools}
\usepackage{thmtools}

\newcommand{\inner}[2]{\left(#1, #2\right)}
\def\bbe{\mathbb{E}}
\usepackage{multirow}

\newcommand{\dt}{\delta}

\usepackage[capitalize]{cleveref}

\hypersetup{
    colorlinks=true,
    citecolor=blue,
    linkcolor=blue,
    filecolor=cyan,
    urlcolor=magenta,
}

\theoremstyle{plain}
\newtheorem{theorem}{Theorem}[section]

\newtheorem{lemma}[theorem]{Lemma}
\newtheorem{corollary}[theorem]{Corollary}
\theoremstyle{definition}
\newtheorem{definition}[theorem]{Definition}

\newtheorem{problem}[theorem]{Problem}
\theoremstyle{remark}
\newtheorem{remark}[theorem]{Remark}

\def\CO{\mathcal{O}}
\newcommand{\bd}[1]{\boldsymbol{#1}}

\newcommand{\comment}[1]{}

\usepackage[textsize=tiny]{todonotes}

\title{Efficient Quantum Algorithms for Quantum Optimal Control}

\author[*]{Xiantao Li}

\author[$\dag$]{Chunhao Wang}

\affil[*]{Department of Mathematics, Pennsylvania State University}
\affil[$\dag$]{Department of Computer Science and Engineering, Pennsylvania State University}
\affil[*]{Email: xiantao.li@psu.edu}
\affil[$\dag$]{Email: cwang@psu.edu}
\date{}

\begin{document}

\maketitle
\begin{abstract}
  In this paper, we present efficient quantum algorithms that are exponentially faster than classical algorithms for solving the quantum optimal control problem. This problem involves finding the control variable that maximizes a physical quantity at time $T$, where the system is governed by a time-dependent Schr\"odinger equation. This type of control problem also has an intricate relation with machine learning. Our algorithms are based on a time-dependent Hamiltonian simulation method and a fast gradient-estimation algorithm. We also provide a comprehensive error analysis to quantify the total error from various steps, such as the finite-dimensional representation of the control function, the discretization of the Schr\"odinger equation, the numerical quadrature, and optimization. Our quantum algorithms require fault-tolerant quantum computers.

\end{abstract}

\section{Introduction}
  Optimal control is an important application of machine learning and optimization. 
Meanwhile, it has been demonstrated in recent works that important insights on deep learning 
can be obtained by interpreting the process of training a 
deep neural network as a discretization of an optimal control  problem  \cite{Weinan2019,he2016deep}.
  In recent decades, following the emergence of quantum computers, the optimal control of quantum systems has attracted considerable attention. This is because quantum properties are largely responsible for many of the recent developments in material science and chemical engineering, and such properties are best utilized with external controls. 
\cite{geppert2004laser,datta2005quantum,mccreery2013critical}. Quantum optimal control (QOC) algorithms \cite{brif_control_2010,werschnik_quantum_2007}, which uses first principle-based computer simulations to identify the desired control variable, has been an important route toward this direction. 
QOC problems have also emerged  in recent quantum computing and quantum information technology. For example,  QOC has recently been shown to have remarkable connections and applications to machine learning techniques using quantum circuits, including variational quantum  algorithms (VQA) \cite{choquette2021quantum,yang2017optimizing},  quantum neural networks~\cite{larocca2021theory}, and quantum approximate optimization algorithm (QAOA)~\cite{farhi2014quantum,cerezo2021variational,LPQ+22}. QOC is treated as a particular machine learning application in \cite{banchi2021measuring}. One specific application of QOC to quantum machine learning is the diagnostics of quantum barren plateau \cite{larocca2022diagnosing}. 
In addition, solving the optimal control problems for quantum systems is a critical step for implementing quantum computers.

The problem of controlling a quantum system can be formulated based on the time-dependent Schr\"{o}dinger  equation (TDSE),
\begin{equation}\label{eq: tdse}
  i \partial_t \ket{\psi} = \big( {H}_0 -u(t) \mu \big) \ket{\psi}, \ket{\psi(0)}= \ket{\psi_0},  
\end{equation}
with an external control variable $u(t)$ that enters the Hamiltonian  through an operator $\mu.$
One example comes from electron transport, where $ H_0 = -\frac{\nabla^2}2 + V(x)$  and $\mu$ is the dipole operator. To keep our discussions focused on the quantum algorithms, we consider the case with a single control. The treatment of multiple control variables and non-unitary dynamics will be presented in separate works.  
We assume the copies of the initial state $\ket{\psi_0}$ are given \emph{for free}, i.e., we have access to the unitary $U_{\psi_0}$ that prepares $\ket{\psi_0}$ as $U_{\psi_0}\ket{0} = \ket{\psi_0}$. 
 
The goal of QOC with time duration $T$ is to determine a control variable $u(t)$ for $t \in [0, T]$, such that a physical quantity, represented by a Hermitian operator $ O,$ is maximized in that
\begin{equation}\label{eq: qctr}
  \max_{ u(t) } J(u), \quad J(u) \coloneqq \langle \psi(T) |  O | \psi(T) \rangle -  \alpha \int_0^T |u(t)|^2 \dd t.
\end{equation} 

Here $\alpha>0 $ can be regarded as a regularization parameter imposing a soft constraint on  the magnitude of the control variable $u(t) $. 
In the context of optimization, the function $J(u)$ is the objective (fitness) function. Aside from the generalization to multiple control variables,  another straightforward extension is to control signals with explicit parametric forms, e.g., a trigonometric expression $u(t)= \sum_{j} E_j \cos(\omega_j t) e^{-(t-t_j)^2/\sigma_j }$ modulated by a Gaussian envelope, which can be written in an abstract form as $u(t)= F(t,\bm \theta) $ with $\theta$ embodying all the parameters. Then the objective function in \cref{eq: qctr} can be optimized over $\bm \theta$ and the gradient can be trivially obtained using chain rule: $\partial J/\partial \bm \theta = (\partial J/\partial u)\, \partial F/\partial \bm \theta.$
By restricting the control variable $u(t)$ to square-integrable functions, the objective function is clearly upper bounded. In addition, it is lower semicontinuous, which then implies that  a maximum exists \cite{ciaramella_newton_2015}. But in practice, the function is often nonconvex and  possesses multiple~maxima.

In addition to the wide range of applications of QOC, there is a host of mathematical analysis and computational techniques based on classical algorithms  \cite{castro_controlling_2012,werschnik_quantum_2007,eitan_optimal_2011,glaser_training_2015,ciaramella_newton_2015,hellgren_optimal_2013,dalessandro_topological_2000,dalessandro_introduction_2008,albertini_lie_2002}. The present paper, however, will be concerned with computer simulations using quantum algorithms, where the $n$-qubit Hamiltonian operator $ H$ (resp.~$\mu$) is represented as a $2^n$-dimensional matrix $H\in \mathbb{C}^{2^n\times 2^n}$ (resp.~$\mu\in \mathbb{C}^{2^n\times 2^n}$). 
In this paper, we assume the Hamiltonians $H$, $\mu$, and the observable $O$ are $d$-sparse: there are at most $d$ nonzero entries in each row/column. In other words, these matrices are \emph{efficiently row/column computable}. To access these matrices, we assume we have access to a procedure $\mathcal{P}_A$ for a matrix $A$ that can perform the following mapping:
  \begin{align}
    \ket{i}\ket{0} \mapsto \ket{i}\ket{r_{i,k}},
  \end{align}
  where $r_{i,k}$ is the $k$-th nonzero entry of the $i$-th row of $A$. In addition, $\mathcal{P}_A$ can also perform the following mapping:
  \begin{align}
    \ket{i}\ket{j}\ket{0} \mapsto \ket{i}\ket{j}\ket{A(i, j)}.
  \end{align}
  Note that it is easy to adapt our algorithms to other more general input models, such as the block-encoding \cite{LC19,CGJ19}.

  With quantum access to these matrices, we aim to determine $u(t)$ for all $t$ such that the objective function defined in \cref{eq: qctr} is maximized. However, due to the nonconvex landscape of the objective function, finding the global maximum is not realistic. Instead, we consider the notion of $\epsilon$-first- or second-order convergence conditions (see \cref{1-cond,2-cond}) for the objective function. More specifically, we aim to solve the following problem.
  \begin{problem}
    \label{prob:qoc}
    Consider the $n$-qubit QOC in \cref{eq: qctr} for time duration $T$. Assume\footnote{More generally, if $\norm{H_0}, \norm{\mu} \leq \Lambda$, the overall complexity of our algorithm will pick up an extra $\CO(\Lambda)$ factor because of \cref{eq:l1-norm-H} and the fact that $\norm{A}_{\max} \leq \norm{A}$ for a matrix $A$.} $\norm{H_0}, \norm{O}, \norm{\mu} \leq 1$, and $\alpha \geq 2/T$. Let $N$ be the number of time intervals and define $\delta \coloneqq T/N$. Suppose $H_0$, $\mu$ and the observable $O$ are $d$-sparse, and suppose we are given access to $U_{\psi_0}$ for preparing $\ket{\psi_0}$. Given the sparse access $\mathcal{P}_{H_0}$ to $H_0$, $\mathcal{P}_{\mu}$ to $\mu$, and $\mathcal{P}_{O}$ to $O$, determine $\tilde{u}(t_j)$ ($t_j = j\delta$) for all $j \in [N]$ such that $\big|\tilde{u}(t_j) - u(t_j)\big|\leq \epsilon$, 
    where $u(t)$ is an $\epsilon$-first-order stationary point of $J(u)$ satisfying \cref{1-cond} or an $\epsilon$-second-order stationary point satisfying \cref{2-cond}.
  \end{problem}

The main challenge in solving \cref{prob:qoc} is due to the repeated solution of TDSE in \cref{eq: tdse}, which becomes unfavorable for classical computers because the dimension grows exponentially. 
In addition, many optimization algorithms require the computation of the derivatives of the objective function with respect to the control variable $u(t).$ A direct calculation of such derivatives of $J$ in \cref{eq: qctr} also requires visiting the TDSE.  The most well-known classical method for optimizing \cref{eq: qctr} is due to Zhu and Rabitz \cite{zhu1998rapid}. The idea is to regard the TDSE in \cref{eq: tdse} as a constraint, which is incorporated into the objective function using a Lagrange multiplier.
\comment{
\begin{equation}\label{eq: J}
\begin{aligned}
   J[u, \psi, \chi]=&  \bra{\psi(T) }\hat O \ket{\psi(T)} -  \alpha \int_0^T |u(t)|^2 dt
     \\&+  \int_0^T \bra{\chi(t)}{\big(\partial_t + i  H(t)\big)}\ket{\psi(t)}  dt.  
\end{aligned}
\end{equation}
Here we have defined  $ H= H_0 -u(t) \mu$.  In order for $J[u, \psi, \chi]$ to achieve a maximum, the derivative with respect to $\chi$ has to be zero. This necessary condition implies that the TDSE in \cref{eq: tdse} is automatically satisfied, in which case the last term in $J[\psi, \chi]$ drops out and the problem is reduced to \cref{eq: qctr}.}
For this optimization problem, algorithms can be constructed by enforcing the first-order condition.
Zhu and Rabitz  \cite{zhu1998rapid} established the monotone convergence property of this algorithm.
This approach is later extended to further improve the convergence properties \cite{maday_new_2003,eitan_optimal_2011,ciaramella_newton_2015} including formulations based on the Liouville von Neumann~equation. 

Nevertheless, the complexity associated with one iteration of these classical algorithms has the scaling of $\CO\big(\mathrm{poly}(2^n) T \big)$, where $n$ is the number of system qubits and we did not include $\epsilon$ in the complexity. The exponential dependence on $n$ comes from the fact that at each time step, the wave functions have to be computed from previous time steps using an approximation of the unitary operator. Depending on the specific methods, this may involve matrix multiplications or solutions of a linear system of equations \cite{castro2004propagators}. The exponential scaling significantly limits the scope of the application, since $2^n$ is usually proportional to the domain size and the number of electrons.

\noindent
\textbf{Main Contributions.}
In this paper, we present \emph{efficient} quantum algorithms for solving the quantum optimal control problem defined in \cref{eq: qctr} that scale polynomially in $n$ --- our quantum algorithms have \emph{exponential speedups} compared with best known classical algorithms.  
One key contribution is a complexity analysis that takes into account various error contributions, including ({\bf i}) the finite-dimensional approximation of $u(t)$, ({\bf ii}) the discretization of the objective function in \cref{eq: qctr}, ({\bf iii}) the approximation of the wave function, and ({\bf iv})  the optimization error. Such comprehensive overall  estimates  are extremely valuable to assess the resources needed for a specific QOC problem. We start with the Hamiltonian simulation method for time-dependent problems \cite{berry2020time} to access the objective function, thus enabling a quantum gradient estimation \cite{gilyen2019optimizing} approach, so that a gradient-based optimization method can be applied. Although our approach is formulated for the TDSE, it can be readily generalized to other control problems where the dynamics is not Hamiltonian, e.g., ODEs implemented with quantum spectral method \cite{childs2019quantum}.

The following theorem summarizes our main result, where the $\widetilde{\mathcal{O}}$ notation has omitted logarithmic factors.

\begin{theorem}[Summarizing \cref{thm:main-c2} and \cref{thm:main-smooth}]
  Assume that $u(t) \in \mathrm{C}^2([0,T]). $
There exists a quantum algorithm that, with probability at least $2/3$\footnote{Using standard techniques, the success probability can be boosted to a constant arbitrarily close to 1 while only introducing a logarithmic factor in the complexity.}, solves \cref{prob:qoc} using
      $\mathcal{\widetilde{O}}\left(\frac{d LT^{2} }{\epsilon^{3}}\right)$
  queries to $\mathcal{P}_{H_0}$ and $\mathcal{P}_{\mu}$, and 
    $\widetilde{\CO}\left(\frac{dnL T^2 }{\epsilon^{3}} + \frac{LT^{3/2}}{\epsilon^{5/2}}\right)$
  additional 1- and 2-qubit gates.

Further, when the control variable $u(t)$ is smooth, the gate complexity becomes 
    $\widetilde{\CO}\left(\frac{dnL T^2 }{\epsilon^{3}} + \frac{LT}{\epsilon^2}\right)$. The Lipschitz constant $L$ can be bounded by \cref{eq: L}
\end{theorem}

Although the smooth control variables does not improve the gate complexity asymptotically, it can reduce the number of gates in practical implementations. We also use \cref{tab:comparison} for a clear comparison of different smoothness assumptions on the control variable $u(t)$.

\begin{table}[ht]
\begin{center}
\begin{tabular}{|c|c|c|}
    \hline
    Smoothness & Queries & Additional gates \\    
    \hline
    Without smoothness &  $\mathcal{\widetilde{O}}\left(\frac{{dL T^{2}}}{\epsilon^{3}}\right)$ & $\widetilde{\CO}\left(\frac{dnL T^2 }{\epsilon^{3}} + \frac{LT^{3/2}}{\epsilon^{5/2}}\right)$\\
    \hline 
    Smooth control  &  $\mathcal{\widetilde{O}}(\frac{dL {T}^{2}}{\epsilon^3})$ & $\widetilde{\CO}\left(\frac{d nLT^2 }{\epsilon^{3}}+\frac{LT}{\epsilon^2}\right)$\\
    \hline
\end{tabular}
\end{center}
\caption{The query and gate complexities for solving \cref{prob:qoc}.}
\label{tab:comparison}
\end{table}

It is important to note that our algorithms require fault-tolerant quantum computers, and hence they are not suitable for near-term quantum devices.

\noindent {\bf Related work.}  To the authors' knowledge, the first attempt to solve the quantum control problem using quantum algorithms has been proposed by \cite{li2017hybrid}. This approach requires  a classical device  to update the control variable $u(t)$ using a gradient-based method.
More recently, Brif et al. \cite{brif_control_2010} considered a specific application, where a molecular wave function is evolved under a laser pulse. The algorithm is also hybrid in nature, utilizing a classical device to update the control variable. The approach outputs the objective function, rather than the gradient, for the Nelder-Mead optimization method. None of the above-mentioned works uses the full power of quantum computation, including the state-of-the-art quantum algorithms for Hamiltonian simulation algorithms and gradient estimation. More importantly, precise query complexity that takes into account all levels of approximations is not provided, making it difficult to assess the applicability of the algorithms to specific quantum control problems in practice. 

The QOC problems show great similarity with variational quantum algorithms (VQA) in general \cite{cerezo2021variational}, except that in most VQAs, the unitary operations are associated with parameterized gates, rather than time-dependent Hamiltonians in QOC, where the smoothness of the control function also plays a role in the error. 
Several other algorithms have recently been proposed to help with the computation of the gradient in VQA \cite{crooks2019gradients,wierichs2022general,kyriienko2021generalized}. 

Very recently, Leng, Peng, Qiao, Lin, and Wu~\cite{LPQ+22} proposed a differentiable programming framework for quantum computers to evaluate gradients of time-evolved states. One of their application is QOC. In their model, the control function $u(t)$ is fixed while leaving freedoms to the parameters of $u$ to be tuned. In contrast, we work on a more general control model and deal with the control function in real time directly. In addition, our work focuses on quantum algorithms on fault-tolerant quantum computers, and we use state-of-the-art techniques to improve the complexity.

\noindent
\textbf{Notation.} We provide a brief introduction to quantum computation in \cref{sec:qcintro}. Here, we summarize the notations used in this paper. First of all, we should distinguish quantum states from a vector in the context of classical computation, although a quantum state is mathematically represented as a vector. In the context of quantum dynamics, we use the Dirac notation $\ket*{\cdot}$ to represent quantum states. For vectors in the context of classical computation, e.g., the gradient vector and the control variable at multiple time steps, we use the bold font, e.g., $\bd{u}$. For a positive integer $k$, we use $[k]$ to denote the set $\{1, \ldots, k\}$. 
 For a vector $\bd{u}$, we use the subscript with a normal font to denote its entries, i.e., $u_1, \ldots, u_d$ are the entries of $\bd{u}$. 
For a vector $\bd{u}$, we use $\norm{\bd{u}}$ to denote its Euclidean norm, i.e., $\norm{\bd{u}} = \sqrt{\sum_{j=1}^d|u_j|^2}$. We use $\norm{\bd{u}}_{\infty}$ to denote its maximum norm, i.e., $\norm{\bd{u}}_{\infty} = \max_j |u_j|$.

In addition, we denote by $\mathrm{C}^2[0, T]$ the class of twice continuously differential functions defined in $[0, T]$.  For a function $f$, we use $\norm{f}_2$, or $\norm{f(t)}_2$ to denote its $L_2$ norm.
Namely, 
$\norm{f(t)}_2 = \norm{f}_2= \int_0^T f(t)^2\, \dd t$.

\section{Preliminaries}

 To solve \cref{prob:qoc}, we need to deal with several levels of approximations, including a finite-dimensional representation of the function $u(t)$, the time discretization of the TDSE in \cref{eq: tdse}, the approximation of the gradient of $\bra{\psi(T) }  O \ket{ \psi(T) }$, and the optimization error.  In this section, we first give a simplistic yet sufficient introduction to quantum computation. Then we analyze the error caused by the discretization of $u(t)$, the discretization of the TDSE, and optimization in subsequent subsections. 

 \subsection{A brief introduction to quantum computation}
 \label{sec:qcintro}

The most daunting obstacle for learning quantum computing is probably its notations. Once one is getting familiar with the Dirac notation, quantum computing can be understood as matrix-vector multiplications. For an $n$-qubit state, we use a $2^n$-dimensional complex (column) vector with unit Euclidean norm to mathematically describe it, and this vector is denoted, in the Dirac notation, by $\ket{\psi}$. We use $\bra{\psi}$ to denote the conjugate transpose of $\ket{\psi}$, and then the inner product between two states can be conveniently written as $\braket{\psi}{\phi}$. When the symbol in the Dirac notation is a natural number, i.e., $\ket{0}, \ldots, \ket{2^n-1}$, we use such states to specifically denote the set of the \emph{standard basis} of a $2^n$-dimensional vector space, same as the basis vectors $\bm{e}_0, \ldots, \bm{e}_{2^n-1}$ that we usually see in linear algebra textbooks. With respect to the standard basis, let $\alpha_0,  \ldots, \alpha_{2^n-1}$ be the entries of $\ket{\psi}$, i.e., $\ket{\psi} = \bigl[\alpha_0,\ldots,\alpha_{2^n-1}\bigr]^T$. It can be conveniently written as a \emph{superposition}, $\ket{\psi} = \alpha_0\ket{0} + \cdots + \alpha_{2^n-1}\ket{2^n-1}$. When $\ket{\psi}$ is \emph{measured}, with probability $|\alpha_j|^2$, it will \emph{collapse} to $\ket{j}$. Note that a measurement is an irreversible process: a state cannot be restored once it collapses to some basis state. 

The building blocks of quantum algorithms are 1- and 2-qubit gates: they are modeled as $2\times 2$ (1-qubit), or $4\times 4$ (2-qubit) unitary matrices. A network of quantum gates form a quantum circuit, which, acting on $n$ qubits, is  mathematically represented as a $2^n\times 2^n$ unitary matrix. In essence, a quantum algorithm starts from an $n$-qubit initial state $\ket{0}$ that is easy to prepare, and applies a quantum circuit, represented by $U$, yielding the state $\ket{\psi} = U\ket{0}$. Then a measurement is performed on $\ket{\psi}$ to obtain the desired outcomes. Readers may refer to~\cite{NC02} for a thorough introduction to quantum computation.

\subsection{Approximation Error of the Objective Function} 
We first represent the control variable $u(t)$ at finitely many time steps with the length of a time interval denoted by $\delta$: Let  $t_j=j\dt $ be discrete time steps,  and $N=T/dt$ be the number of time steps. Then we represent  $u(t)$ via piece-wise linear interpolation, denoted by $u_I(t)$, given by,
\begin{equation}\label{uI}
   u_I(t) = \sum_{j=0}^N u_j b_j(t), 
\end{equation}
with $b_j(t)$ being the nodal basis functions (i.e., the hat functions) centered at time steps $t_j$.
The interpolating function can also be integrated directly, which leads to a quadrature approximation for the integral in \cref{eq: qctr}. The error from these approximations  follows   the standard error bound, e.g., see \cite{kincaid2009numerical}. 
\begin{theorem}[{\bf Perturbation under $\mathrm{C}^2$ assumption.}]
Assume that $u(t) \in \mathrm C^2([0,T]),$ i.e., it is twice continuously differentiable, then the error from the piecewise-linear interpolation is bounded as follows,
\begin{equation}
\| u(t) - u_I(t) \|_2  \leq C \| u''(t) \|_2  \delta^2.
\end{equation}  
In addition, the integral of $u(t)^2$ can be approximated by a composite trapezoid rule, with weights denoted by $w_j$,
\begin{equation}\label{eq: q-err}
    \int_0^T u(t)^2 \,\dd t = \delta  \sum_{j=0}^N w_j u(t_j)^2  + \mathcal{O}\left(T \delta^2  \right).
 \end{equation}
\end{theorem}
The quadrature formula ($w_0=w_N=1/2$ and $w_1=w_2=\cdots=w_{N-1}=1$) can be derived from the linear interpolation of $u(t)$ in each time interval by integrating the nodal functions directly. In addition, the error carries a prefactor that is proportional to the second-order derivatives of the integrand, which is bounded due to the assumption that $u(t) \in \mathrm C^2([0,T]).$

With this representation of the control function, the problem is reduced to a finite-dimensional problem. We will use $\bd{u}$ to represent the nodal values of the interpolation: $\bd u=\big\{u_j \big\}_{j=1}^N$, with $u_j$ being the nodal value of $u(t)$ at $t_j.$ We will then write the corresponding objective function as $\widetilde{J}(\bd{u}),$  a function with $N$ variables. Namely, 
\begin{equation}\label{fobj'}
    \widetilde{J}(\bd{u}) =  \bra{\psi_N } O \ket{\psi_N}  - \alpha \delta \sum_{j=0}^N w_j u_j^2. 
\end{equation}
Here we have also included the approximation of $\ket{\psi(T)}$, denoted by $\ket{\psi_N}$.

\begin{corollary}\label{eq: eps2N}
 Assume that the  wave function is approximated by $\ket{\psi_N}$, with an error comparable to the quadrature error in \cref{eq: q-err}, i.e., 
 \begin{equation}
     \norm{\ket{\psi_N} - \ket{\psi(T)}} = \CO( T \delta^2).
 \end{equation}
Further assume that the maximum of $J[u]$ is achieved at some $u^*(t) \in \mathrm C^2([0,T])$, and the maximum is  not degenerate. Then for sufficiently small $\epsilon,$ by choosing,
\begin{equation}\label{eq: N-est}
    N= \CO\left({\frac{T^{3/2}}{\epsilon^{1/2}}}\right),
\end{equation}
the approximate objective function $\widetilde{J}(\bm{u})$ 
has a maximum $\bm{u}^*,$ for which the corresponding interpolant satisfies,
\[ \norm{ u(t) - u_I(t) }_2 \leq \epsilon.\]
\end{corollary}
\begin{proof}
  This can be proved by using the continuous dependence of the wave function $\ket{\psi(t)}$ on the control variable $u(t)$, e.g., see Proposition 4 of \cite{ciaramella_newton_2015}. 
\end{proof}

\subsection{Optimization method} 
A gradient-base algorithm updates the control variables using the gradient  \cite{nesterov2018lectures} to generate an approximating sequence $\{\bd{u}^{(k)}\}_{k\geq 0}$, 
\begin{equation}\label{gd}
    \bd{u}^{(k+1)}  = \bd{u}^{(k)} + \eta \big( \nabla J(\bd{u}^{(k)}) + \bd{\zeta}_k \big).
\end{equation}
{Notice that the QOC in \cref{eq: qctr} is customarily formulated as a maximization problem. Therefore, the algorithm updates the control variables along the gradient \emph{ascent} direction. In addition, we considered the perturbed gradient method \cite{jin2021nonconvex}, which introduced a Gaussian noise $\bd{\zeta}_k$ with mean zero and $\bbe[\bd{\zeta}_k^2]=r^2$ for some $r>0$ to help the solution with fulfilling an approximation second-order condition.}

The convergence to the global maximum of a gradient-based algorithm requires certain convexity assumptions on $J$, e.g., the star-convexity \cite{zhou2018sgd}, which, however, may not be easy to verify in practice.      
Alternatively, one can seek approximate optimality conditions.  Jin et al.~\cite{jin2021nonconvex} introduced  $\epsilon$-first and second order conditions, 
\begin{align}\label{1-cond}
  \norm{\nabla J(\bm{u}) } &< \epsilon, \; \text{and} \\
  \label{2-cond}
  \nabla^2 J(\bm{u}) &\leq \sqrt{\rho \epsilon} I.
\end{align}

Here $\rho$ is the Lipschitz constant of the Hessian. 
The following result, adapted from Theorem 4.1 of \cite{jin2021nonconvex} (with $\sigma=\mathcal{O}(\epsilon^2)$ ) summarizes the complexity associated with optimization.

\begin{theorem}[Adapted from Theorem 4.1 in \cite{jin2021nonconvex}] 
\label{thm:optimization}
Assume that the gradient $\nabla J$ is $L-$Lipschitz and it can be estimated with an unbiased estimation of variance $\sigma^2=\CO{(\epsilon^2)}$ and that the learning rate is $\eta< \Theta\left(\frac{1}{L}\right)$. With probability at least $1-\delta$, the optimization method defined in \cref{gd}   will visit an approximate stationary point that fulfills the first order conditions in \cref{1-cond},  at least once within the following number of iterations,
\begin{equation}\label{eq: k-cmpl}
 k=\widetilde{\mathcal{O}}\left( \log \frac{1}{\delta} L \frac{ J(\bd u^* ) - J(\bd{u}^{(0)}) }{\epsilon^2} \right).
\end{equation}
\end{theorem}

For the objective function in \cref{fobj'}, the Lipschitz constant can be estimated using the derivative bounds from \cref{lemma:partial-alpha},
\begin{equation}\label{eq: L}
    L \leq T \delta \norm{\mu}^2 + 2 \alpha \delta.  
\end{equation}
Since we do not assume an upper bound on $\alpha$, we will treat $L$ as a separate constant.

\section{Quantum algorithm and complexity analysis}
\label{sec:hamtoniansimulation}

In order to access the objective function in  \cref{eq: qctr}, a time discretization is needed to approximate the integral and the wave function $\psi(T)$  by numerical integrators, as alluded to in the previous section.   To begin with, we write $\widetilde{J}$ as two terms:
  $\widetilde{J}(\bm{u}) = \widetilde{J}_1(\bm{u}) + \widetilde{J}_2(\bm{u})$,
where
\begin{align}
  \label{eq:j1}
  \widetilde{J}_1(\bm{u}) &\coloneqq \bra{\psi_N}O\ket{\psi_N} \quad\text{ and}\\
  \label{eq:j2}
  \widetilde{J}_2(\bm{u}) &\coloneqq \alpha\delta\sum_{j=1}^Nw_ju(t_j)^2. 
\end{align}
Note the $\widetilde{J}_2$ is easy to compute on a classical computer. Thus we only focus on $\widetilde{J}_1$, and use a quantum algorithm to speed up its calculation.

\subsection{Obtaining the final state}
To have quantum access to $\widetilde{J}_1$, we need to efficiently prepare of $\ket{\psi_N}$ via a time-marching procedure, i.e.,
\begin{equation}\label{eq: HS}
 \ket{ \psi_{N} } = V(t_{N},t_{N-1}) V(t_{N-1},t_{N-2}) \cdots V(t_1,t_0) \ket{ \psi_{0} }.
\end{equation}
Note that $t_N=T$. At each step, the dynamics is evolved via an operator $V(t_{j},t_{j-1})$.  For instance, in the Dyson series approach \cite{low2018hamiltonian,kieferova2019simulating}, the operator $V$ can be constructed as follows,
\begin{align}
  \begin{aligned}
    V(&t_{j+1},t_j)  = \\
                    &\sum_{k=0}^K \frac{(-i\delta)^k}{M^k k!} \sum_{j_1,j_2,\cdots,j_k=0}^{M-1} \mathcal{T} H(s_{j_k})\cdots H(s_{j_1}).
  \end{aligned}
\end{align}
Here the symbol $\mathcal{T}$ indicates that only those points where $s_{j_1}\leq s_{j_2} \leq \cdots \leq s_{j_k}$ are collected. The points are selected to approximate the integrals in the Dyson series expansion \cite{berry2017quantum}. The complexity estimate was  presented in Theorem 9 of \cite{low2018hamiltonian}. We use a more efficient simulation algorithm by Berry, Childs, Su, Wang, and Wiebe~\cite{berry2020time} where the dependence on $T \max_t \norm{H(t)}_{\max}$ has later been relaxed to an averaged $L_1$ norm in time. 

\begin{lemma}[\cite{berry2020time}]
  \label{thm:hamiltoniansimulation}
  Assume that $n$-bit Hamiltonian $H(t)$ is $d$-sparse for all $t \in [0, T]$. The TDSE in \cref{eq: tdse} can be simulation until time $T$ within error $\epsilon$ and failure probability $\CO(\epsilon)$ using 
\begin{equation}
  \mathcal{O}\left( d \int_0^T\, \dd \tau \norm{H(\tau)}_{\max} \frac{\log (d H_\mathrm{max} T/\epsilon) }{\log \log (d H_\mathrm{max} T/\epsilon)} \right),
\end{equation}
queries to the oracle of $H(t)$ 
and $\widetilde{\CO}(dn\int_0^T\, \dd \tau \norm{H(\tau)}_{\max})$ additional 1- and 2-qubit gates.
\end{lemma}

Once $\ket{\psi_N}$ from \cref{eq: HS} is prepared by \cref{thm:hamiltoniansimulation}, we can use it to estimate the gradient of $\widetilde{J}$ defined in \cref{fobj'}. We take advantage of quantum gradient estimation algorithms to estimate $\nabla \widetilde{J}$. These algorithms require certain quantum access to the function. In particular, we define the two commonly used oracles, namely, the \emph{probability oracle} and the \emph{phase oracle} as follows.

\begin{definition}[Probability oracle]
\label{defn:proboracle}
Consider a function $f:\mathbb{R}^d\rightarrow [0,1]$. The \emph{probability oracle} for $f$, denoted by $U_f$, is a unitary defined as
\begin{align}
    &U_f\ket{\bm{x}}\ket{\bm{0}} = \\
    &\ket{\bm{x}}\left(\sqrt{f(\bm{x})}\ket{1}\ket{\phi_1(\bm{x})} + \sqrt{1-f(\bm{x})}\ket{0}\ket{\phi_0(\bm{x})} \right), \nonumber
\end{align}
where $\ket{\phi_1(\bm{x})}$ and $\ket{\phi_0(\bm{x})}$ are arbitrary states.
\end{definition}

\begin{definition}[Phase oracle]
\label{defn:phaseoracle}
Consider a function $f:\mathbb{R}^d\rightarrow \mathbb{R}$. The \emph{phase oracle} for $f$, denoted by $O_f$, is a unitary defined as
\begin{align}
    O_f\ket{\bm{x}}\ket{\bm{0}} = e^{if(\bm{x})}\ket{\bm{x}}\ket{\bm{0}}.
\end{align}
\end{definition}

The following result from Theorem 14 of~\cite{gilyen2019optimizing} shows an efficient conversion from the probability oracle to the phase oracle.
\begin{lemma}
\label{lemma:prob-to-phase}
Consider a function $f: \mathbb{R}^d \rightarrow [0, 1]$. Let $U_f$ be the probability oracle for $f$. Then, for any $\epsilon \in (0, 1/3)$, we can implement an $\epsilon$-approximate of the phase oracle $O_f$ for $f$, denoted by $\widetilde{O}_f$, such that
    $\norm{\widetilde{O}_f\ket{\psi}\ket{\bm{x}} - O_f\ket{\psi}\ket{\bm{x}}} \leq \epsilon$,
for all state $\ket{\psi}$. This implementation uses $\CO(\log(1/\epsilon))$ invocations to $U_f$ and $U_f^{\dag}$, and $\CO(\log\log(1/\epsilon))$ additional qubits.
\end{lemma}

We will show how to construct the required oracle for $\widetilde{J}_1$ given the circuit that prepares $\ket{\psi_N}$ in the proof of \cref{thm:main-c2}. Before that, let us focus on higher-order derivatives of $\widetilde{J}_1$, which are crucial in quantum algorithms for estimating the gradient of $\widetilde{J}_1$.

\subsection{High-order derivatives of the objective function}
In this subsection, we bound the magnitude of higher-order derivatives of $\widetilde{J}_1$. In particular, we show the following lemma. The proof can be found in \cref{sec: A1}.

\begin{lemma}
  \label{lemma:partial-alpha}
  Let $\bm{\alpha}  = (\alpha_1, \ldots, \alpha_k) \in [N+1]^k$ be an index sequence\footnote{for a precise definition of an index sequence, see Definition 4 of~\cite{gilyen2019optimizing}}. The derivatives of the control function $\widetilde{J}_1$ with respect to the control variables satisfy:
\begin{equation}
      \begin{aligned}
            \norm{ \frac{\partial^{\bm{\alpha}} \widetilde{J}_1}{\partial u_{\alpha_1} u_{\alpha_2} \cdots u_{\alpha_k} }   } 
             \leq (k+1)! \left(\delta  \|\mu\| \right)^k.
      \end{aligned}
\end{equation}
\end{lemma}

\subsection{Quantum gradient estimation}
In this subsection, we show how to efficiently estimate the gradient of $\widetilde{J}_1$. In the following, we state a quantum algorithm for estimating gradients due to Gily{\'e}n, Arunachalam, and Wiebe~\cite{gilyen2019optimizing}, which is an extension Jordan's gradient estimation algorithm~\cite{jordan2005fast}. Note that there exists a nearly-optimal gradient estimation, i.e., Theorem 25 of~\cite{gilyen2019optimizing}; however, this result does not directly apply to our setting since the derivatives, as characterized in \cref{lemma:partial-alpha}, do not satisfy the bound $c^{k}k^{k/2}$. To give a more efficient estimation algorithm than just relying on the bound on the Hessian, we use the following result based on a higher-order finite difference method.
\begin{lemma}[Rephrased from Theorem 23 of~\cite{gilyen2019optimizing}]
  \label{lemma:jordan}
  Suppose the access to $f:[-1, 1]^N\rightarrow\mathbb{R}$ is given via a phase oracle $O_f$. If $f$ is $(2m+1)$-times differentiable and for all $\bm{x}\in[-1, 1]^N$ it holds that
  \begin{align}
    |\partial_{\bm{r}}^{2m+1} f(\bm{x})| \leq B \quad \text{ for $\bm{r} = \bm{x}/\norm{\bm{x}}$},
  \end{align}
  then there exists a quantum algorithm that output an approximate gradient $\bm{g}$ such that $\norm{\bm{g}-\nabla f(\bm{0})}_{\infty} \leq \epsilon$ with probability at least $1-\rho$ using
  \begin{align}
    \widetilde{\CO}\left(\frac{N^{1/2}B^{1/(2m)}N^{1/(4m)}\log(N/\rho)}{\epsilon^{1+1/(2m)}}\right)
  \end{align}
  queries to $O_f$, and $\widetilde{\CO}(N)$ additional 1- and 2-qubit gates.
\end{lemma}

Based on \cref{lemma:jordan}, we have the following result, where the proof is postponed in \cref{sec:proof-grad-est-J}.
\begin{lemma}
  \label{lemma:grad-est-J}
  Let $\widetilde{J}_1$ be defined as in \cref{eq:j1}. Suppose we are given access to the phase oracle $O_{\widetilde{J}_1}$ for $\widetilde{J}_1$. Then, there exists a quantum algorithm that outputs an approximate gradient $\bm{g}$ such that
    $\norm{\bm{g}-\nabla \widetilde{J}_1} \leq \epsilon$
  with probability at least $1-\rho$ using 
    $\widetilde{\CO}\left(T\log(N/\rho)/\epsilon\right)$
  queries to $O_{\widetilde{J}_1}$, and $\widetilde{\CO}(N)$ additional 1- and 2-qubit gates.
\end{lemma}

\subsection{Proof of the main theorem}
We outline the combined algorithm in \cref{alg1}. The bound on $k_{\max}$ follows from \cref{eq: k-cmpl}.

\begin{algorithm}[tb]
 \caption{Quantum algorithm for solving QOC}
\label{alg1}
\SetAlgoLined
Given $T>0$ and $\epsilon>0$, set $N=\frac{T^{3/2}}{\epsilon^{1/2}}$, and $k_{\max} = \mathcal{O}(\frac{L}{\epsilon^2})$ where $L$ is bounded by \cref{eq: L}; set $u(t)=0$ \;

\For{$k=1:k_\mathrm{max}$}{
Use \cref{thm:hamiltoniansimulation,lemma:prob-to-phase} to construct the phase oracle for $\widetilde{J}_1(\bm{u})$\; 
Use \cref{lemma:grad-est-J} to estimate $\bd{g}^{(k)} \approx  \nabla J(\bd{u}^{(k)})$\;
Update control variable: $\bm{u}^{(k+1)} =\bm{u}^{(k)} + \eta  \bm{g}^{(k)} +  \eta \bd{\zeta}_k$\;
}
\end{algorithm}

Before proving the performance of \cref{alg1}, we need to bound the $L_1$ norm $\int_0^T\dd \tau \norm{H(t)}_{\max}$, which depends on the $L_1$ norm of control function. Although in general, no explicit formula is available for $u(t)$, we can deduce some a priori bounds on $u(t)$.
\begin{lemma}
  \label{lemma:l1-norm}
  Suppose that $\norm{O}\leq 1,$ and $\alpha > \frac{2}{T}$. Let $u^{(1)}, \ldots, u^{(k)}$ be a sequence of control variables obtained from each iteration of \cref{alg1}. Let $u^{(1)}(t)=0$ be the initial guess. Then for all $j \in [k]$, the $L_1$ norms of the control functions satisfy
      $\int_0^T \abs{u^{(1)}(t)} \dd t \leq T$.
\end{lemma}
\begin{proof}
  First of all, note that $J(u^{(1)}) \geq -1$. For all $j = 2, \ldots, k$, we have $J(u^{(j)}) \geq J(u^{(1)}) \geq -1$. Again using the fact that $\norm{O}\leq 1$, we find $\int_0^T |u^{(j)}(t)|^2 \dd t \leq 2/\alpha \leq T.$ 
Using Cauchy–Schwarz inequality, we have,
\begin{align}
  \int_0^T \abs{u^{(j)}(t)} \dd t \leq \left[\int_0^T \abs{u^{(j)}(t)}^2 \dd t\right]^{\frac{1}{2}} {T}^{\frac{1}{2}}\leq T. 
\end{align}
\end{proof}  

The inequality $J(u^{(j)}) \geq J(u^{(1)})$ follows from a function ascent property in \cref{Juk-Ju0}, which as we will illustrate in the proof of  \cref{thm:main-c2},  occurs with high probability. 
Because $\norm{H_0}, \norm{\mu} \leq 1$, \cref{lemma:l1-norm} implies that
\begin{align}
  \label{eq:l1-norm-H}
  \int_0^T\, \dd \tau \norm{H(t)}_{\max} = \CO(T).
\end{align}

Now, we are ready to state the main result of this section.

\begin{theorem}
\label{thm:main-c2}
  Assume that $u(t) \in \mathrm{C}^2([0,T]) $. 
  There exists a quantum algorithm that, with probability at least $2/3$, solves \cref{prob:qoc} by visiting a first-order stationary point satisfying \cref{1-cond} using
      $\mathcal{\widetilde{O}}\left(\frac{d LT^{2} }{\epsilon^{3}}\right)$
  queries to $\mathcal{P}_{H_0}$ and $\mathcal{P}_{\mu}$, and 
    $\widetilde{\CO}\left(\frac{dnL T^2 }{\epsilon^{3}} + \frac{LT^{3/2}}{\epsilon^{5/2}}\right)$
    additional 1- and 2-qubit gates\footnote{In the gate complexity, the second term is dominated by the first; however, we keep both to demonstrate how smooth control can improve the gate complexity non-asymptotically.  (See \cref{sec:smooth}.)}. The Lipschitz constant $L$ can be bounded by \cref{eq: L}.
\end{theorem}

\begin{proof}
  Recall that by using \cref{thm:hamiltoniansimulation}, we can obtain $\ket{\psi_N}$ given $\ket{\psi_0}$. Now, we show how to construct the oracle to estimate the gradient of $\widetilde{J}$. 
We first construct the probability oracle $U_{\widetilde{J}_1}$ for $\widetilde{J}_1$ by a Hadamard test circuit. Note that this oracle is well-defined because we have assumed that $\norm{O} \leq 1$. 
In addition, we need access to a block-encoding $U_O$ of $O$, which is a unitary of the form
\begin{align}
  U_O = 
  \begin{bmatrix}
    O & \cdot\\
    \cdot & \cdot
  \end{bmatrix}.
\end{align}
Hence, $\bra{0}\bra{\psi_N}U_O\ket{0}\ket{\psi_N} = \bra{\psi_N}O\ket{\psi_N}$. 
Let c-$U_O$ denote the controlled $U_O$. 
The Hadamard test circuit $(H\otimes I)(\text{c-}U_O)(H\otimes I)$ acting on $\ket{0}\ket{\psi_N}$ produces,
\begin{align}
  \sqrt{f(\bm{u})}\ket{1}\ket{\phi_1(\bm{u}} + \sqrt{1-f(\bm{u})}\ket{0}\ket{\phi_0(\bm{u}},
\end{align}
for
  $f(\bm{u}) \coloneqq -\frac{1}{2}\bra{\psi_N}O\ket{\psi_N} + \frac{1}{2}$.
  Note that we actually constructed the probability oracle for $f(\bm{u})$ instead of $\bra{\psi_N}O\ket{\psi_N}$. However, this is fine because the gradient of the latter is the gradient of the former multiplied by a factor $-1/2$. In addition, we note that the block-encoding $U_O$ can be efficiently constructed using the sparse-access oracles for $O$ (see Lemma 48 of \cite{GSLW19}). It is important to note that the implementation of the time-dependent Hamiltonian simulation by using \cref{thm:hamiltoniansimulation} uses some control variable $\bm{u}$ as part of the input. Such Hamiltonian simulation circuits can be made coherent, i.e., it can be used when the register containing the control variable $\bm{u}$ is in superposition. In addition, we use $U_{\psi_0}$ to prepare the initial state $\ket{\psi_0}$ from $\ket{0}$ for Hamiltonian simulation. Due to the bound on the $L_1$ norm in \cref{eq:l1-norm-H}, we can implement the probability oracle $U_{\widetilde{J}_1}$ as defined in \cref{defn:proboracle} using $\widetilde{\CO}(dT)$ queries to $\mathcal{P}_{H_0}$ and $\mathcal{P}_{\mu}$. We highlight that the oracle for the time-dependent Hamiltonian can be efficiently implemented. The values of the control function away from the nodal points can be easily computed by linear interpolation. Thus, we store $\CO(N)$ parameters in QRAM to construct the input oracle. The gate complexity of the addressing scheme for QRAM is $\widetilde{\CO}(N)$, which is polynomial in $T$ and $\epsilon$ due to \cref{eq: N-est}. The circuit depth of the addressing scheme is $\CO(\log N)$. As a result, considering the gate- and time-complexities of implementing QRAM, the Hamiltonian simulation is still efficient. Further, in special cases of our control models where the form of $u(t)$ is known, while the parameters in $u(t)$ need to be determined, we do not need QRAM because $u(t)$ can be computed by an efficient subroutine. 

Next, we use \cref{lemma:prob-to-phase} to construct an $\epsilon$-approximate phase oracle $\widetilde{O}_f$ defined in \cref{defn:phaseoracle} using $\CO(1)$ queries to $U_{\widetilde{J}_1}$ and hence $\widetilde{\CO}(dT)$ queries to $\mathcal{P}_{H_0}$ and $\mathcal{P}_{\mu}$.

Using \cref{lemma:grad-est-J}, we can obtain a vector $\bm{g}$ such that
    $\norm{\bm{g} - \nabla \widetilde{J}_1/2} \leq \epsilon$
with probability at least $1-\rho$ using $\widetilde{\CO}(T\log(N/\rho)/\epsilon)$ queries to $O_{\widetilde{J}_1}$. Hence, the number of queries to $\mathcal{P}_{H_0}$ and $\mathcal{P}_{\mu}$ is
\begin{align}
\label{eq:qc-forgrad}
    \widetilde{\CO}\left(\frac{dT^2\log(N/\rho)}{\epsilon}\right).
\end{align}

Now we use \cref{thm:optimization}, and analyze the error caused by the estimation error in the spectral norm. We  consider the gradient iteration method with the noise $\bd{\zeta}_k=0$, 
\begin{equation}\label{sgd}
\bd{u}^{(k+1)}  = \bd{u}^{(k)} + \eta \big( \bd{g}_k  +  \bd{\zeta}_k \big).
\end{equation}
The main departure from the optimization algorithm in \cref{gd} is that an approximation  $\bd{g}_k$  of the gradient of a loss function $\nabla J(\bd{u}^{(k)}) $ is used. We denote the error by 
\[
\bd{e}_k:= \bd{g}_k - \nabla J(\bd{u}^{(k)}).
\]
The proof of \cref{thm:optimization} in \cite{jin2021nonconvex} relies on a descent property: Lemma B.1 in \cite{jin2021nonconvex}. Following this approach, we  first show that the following inequality holds with probability at least $1-\nu$,
\begin{equation}\label{Juk-Ju0}
     J( \bd{u}^{(k)} ) \geq J( \bd{u}^{(0)} ) +\frac{\eta}{2} \sum_{j=0}^{k-1}   \norm{ \nabla J(\bd{u}^{(j)})}^2  - \frac{k\eta \epsilon^2}{4}
\end{equation}
where the algorithm in \cref{lemma:grad-est-J} is implemented with $\rho= \nu/k$. 
We leave the derivation of this inequality to \cref{sec: A2}.

In light of \cref{fobj'}, the fact that $w_j=\mathcal{O}(1)$, and the derivative bounds in \cref{lemma:partial-alpha}, we see that the Lipschitz constant of $\nabla J$ is $L=\mathcal{O}((1+\alpha)\sqrt{N}  \delta ).$ 
Now if we choose the number of iterations,
\begin{equation}
     k= 4L \frac{J^* - J( \bd{u}^{(0)} )} {\epsilon^2},
\end{equation}
and suppose that $\norm{ \nabla J(\bd{u}^{(j)})}> \epsilon, \forall j \in [k]$.  Namely, the iterations never reached a critical point. Then the descent property implies that,
    $J( \bd{u}^{(k)} ) - J( \bd{u}^{(0)} ) \geq  \frac{ k \eta \epsilon^2 }{4}$,
with probability $1-\nu$, which would lead to a contradiction
 $J( \bd{u}^{(k)} ) - J( \bd{u}^{(0)} ) > J^* - J( \bd{u}^{(0)} ).$ 
Consequently, the iterations in \cref{sgd} will fulfill the first order condition. The rest of the proof can be completed by following the proof in \cite{jin2021nonconvex}.

We use $\nu$ instead of $\rho$ in \cref{eq:qc-forgrad}. The complexity only depends on $\nu$ logarithmically. As a result, we have the claimed overall query complexity. The gate complexity for each iteration is dominated by the gradient estimation and time-dependent Hamiltonian simulation, i.e., $\widetilde{\CO}(dnLT^2/\epsilon + N) = \widetilde{\CO}(dnLT^2/\epsilon + T^{3/2}/\epsilon^{1/2})$.
\end{proof}

\begin{remark}
  The analysis in \cite{jin2021nonconvex} also considered how the interactions from PSGD escape from saddle points. However, the effect of the bias from  Jordan's algorithm on the iterations near saddles points becomes more subtle and requires further analysis. On the other hand, recent analysis of the quantum control landscape \cite{ge2021optimization} suggests that there are control problems that are free of saddle points. For these problem, our results can be applied directly to solve \cref{prob:qoc} by finding a second-order stationary point satisfying \cref{2-cond}.
\end{remark}

\section{Quantum dynamics with smooth control} 
\label{sec:smooth}

The estimates in the previous sections have been obtained based on the mild assumption that the control function $u(t)$ is only $\mathrm C^2.$  Meanwhile, there are various scenarios where the control function is parameterized using smooth functions \cite{song2022optimizing,machnes2018tunable}. For example, rather than controlling point-wise values of $u(t)$, one can express $u(t)$ as Fourier series and then the control parameters are reduced to Fourier coefficients.  One important implication is that the corresponding solution of the TDSE is smooth as well, in which case,  the algorithms presented in \cref{sec:hamtoniansimulation} can be greatly improved. 
In particular,  both the wave function $\psi(T)$ and the integral of the control function within a time interval can be approximated with arbitrary order $p$. Namely,   
\begin{equation}\label{eq: order-p}
\begin{aligned}
     &\ket{\psi(\delta )} - \ket{\psi_1} = \widetilde{\mathcal{O}}\left(\frac{\dt^{p+1}}{(p+1)!}\right), 
     \\
     &\int_0^\delta  u(t)^2 dt - \delta \sum_j w_j u_j = \widetilde{\mathcal{O}}\left(\frac{\dt^{p+1}}{(p+1)!}\right).
     \end{aligned}
\end{equation} 
for arbitrary order $p$. Without loss of generality, we have expressed the approximations here for the first time step.

For the approximation of the integral,  Gaussian quadrature can be used to obtain maximum accuracy. For the approximation of the wave function, such accuracy has been obtained in the Dyson series approach \cite{kieferova2019simulating} even without this smoothness assumption. 

To ensure that the error from the approximation of $\ket{\psi(T)}$ is within precision $\epsilon$,
we  force the one-step error to be below $\epsilon/N$. Since $p$ can be chosen arbitrarily, the step size $\delta=\CO(1).$ In light of \cref{eq: order-p}, this error bound can be achieved if we choose  $p$ as follows,
\begin{equation}
p  = \Theta\left(\frac{\log N/\epsilon}{\log\log N/\epsilon}\right) = \Theta\left(\frac{\log T/\epsilon}{\log\log T/\epsilon}\right).
\end{equation}

Consequently, the number of time steps $N=T/\delta $ is  improved from \cref{eq: N-est} to
\begin{equation}
  \label{eq:N-smooth}
N= \Theta  \left(T \frac{\log T/\epsilon}{\log\log N/\epsilon}   \right) =\widetilde{ \mathcal{O}}\left(T\right).    
\end{equation}

Overall, the smoothness of the control variable $u(t)$ allows us to use high-order approximations and it reduces the extra gates used in \cref{lemma:jordan} by a factor of $T^{1/2}/\epsilon^{1/2}$. However, the smoothness does not improve the query complexity and the overall gate complexity because the number of queries depend on $T = N\delta$, no matter how small $N$ is.

To summarize this section, we have the following theorem.

\begin{theorem}
\label{thm:main-smooth}
  Assume that $u(t)$ is smooth. 
There exists a quantum algorithm that, with probability at least $2/3$, solves \cref{prob:qoc} by visiting a first-order stationary point satisfying \cref{1-cond} using
      $\mathcal{\widetilde{O}}\left(\frac{d LT^{2} }{\epsilon^{3}}\right)$
  queries to $\mathcal{P}_{H_0}$ and $\mathcal{P}_{\mu}$, and 
    $\widetilde{\CO}\left(\frac{dn LT^2 }{\epsilon^{3}} + \frac{LT}{\epsilon^2}\right)$
  additional 1- and 2-qubit gates. The Lipschitz constant $L$ can be bounded by \cref{eq: L}.
\end{theorem}
\begin{proof}
  The proof is similar to that of \cref{thm:main-c2}. We use the number of steps $N$ as in \cref{eq:N-smooth} to obtain this improved gate complexity.
\end{proof}

\section{Numerical Example}

Here we consider a one-dimensional model similar to \cite{zhu_rapid_1998}. The Hamiltonian $H_0$ is a tri-diagonal matrix corresponding to a three-point discrete Laplacian operator. The operator $\mu$ and $O$ are diagonal,
\[ \mu = \sum_r r e^{-r/r_0} \ketbra{r}{r}, \quad  O = \sum_r \frac{\gamma_0}{\pi} e^{-\gamma_0^2 r^2} \ketbra{r}{r}. \]
 The step size is set to $\dt=0.02.$

In the optimization, we choose the learning rate to be $0.04$. Using  $u(t)=0$ as the initial guess, we apply \cref{gd} for 2000 iterations. \cref{fig:loss} shows that the objective function has reached a  plateau.  
As comparison, we run the algorithms with exactly computed gradient, and perturb gradients with noise. We observe that even with noise perturbations, the gradient-based algorithm still converges to the same maximum. In addition, as shown in \cref{fig:ctrl}, the resulting control variable   exhibits some statistical fluctuations; however, it still remains close to the optimal control variable, indicating its resilience to noisy perturbations.

\begin{figure}[ht]
\vskip 0.2in
\begin{center}
    \centerline{\includegraphics[width=0.8\columnwidth]{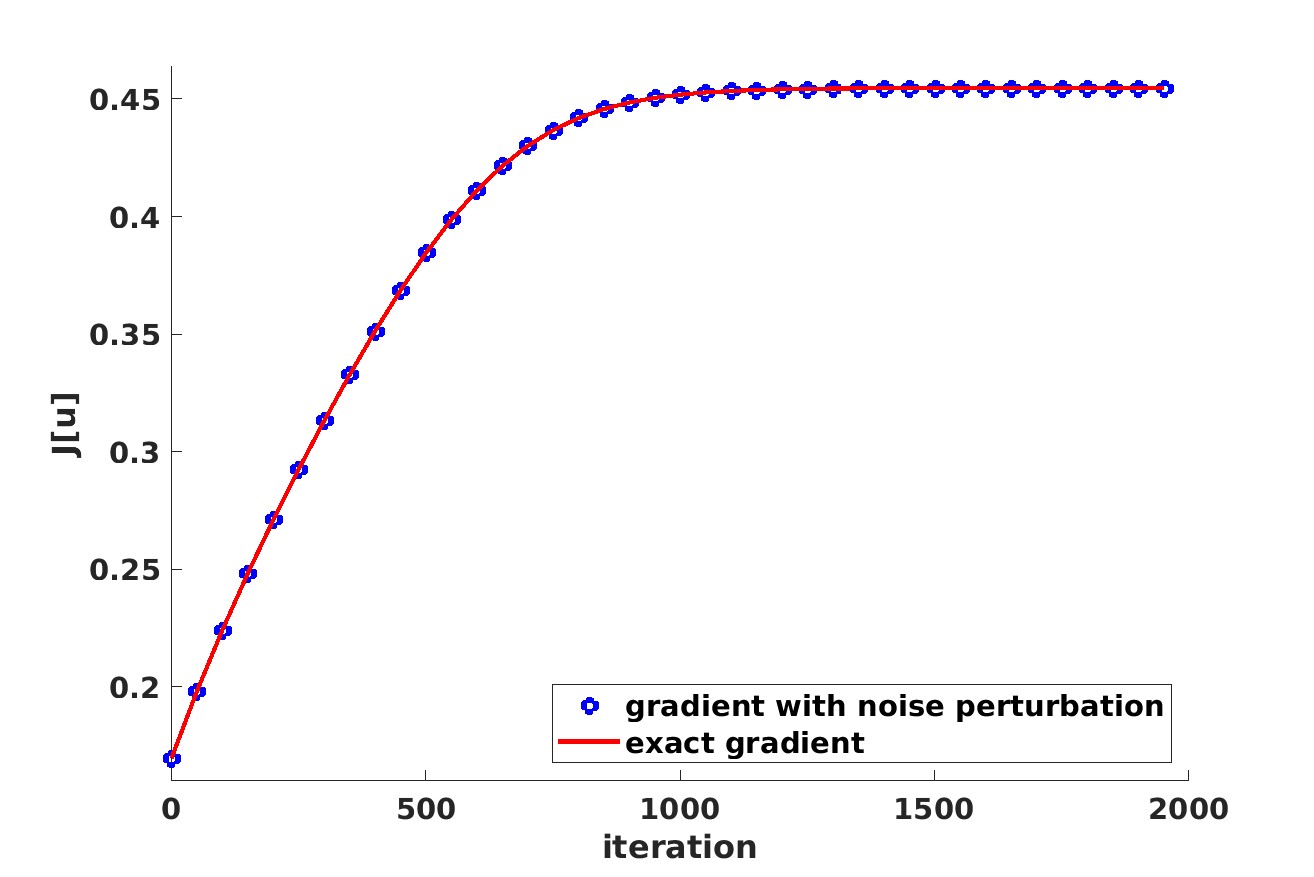}}
    \caption{The loss function during 2000 iterations.}
    \label{fig:loss}
    \end{center}
\vskip -0.2in
\end{figure}

\begin{figure}[ht]
\vskip 0.2in
\begin{center}
    \centerline{\includegraphics[width=0.8\columnwidth]{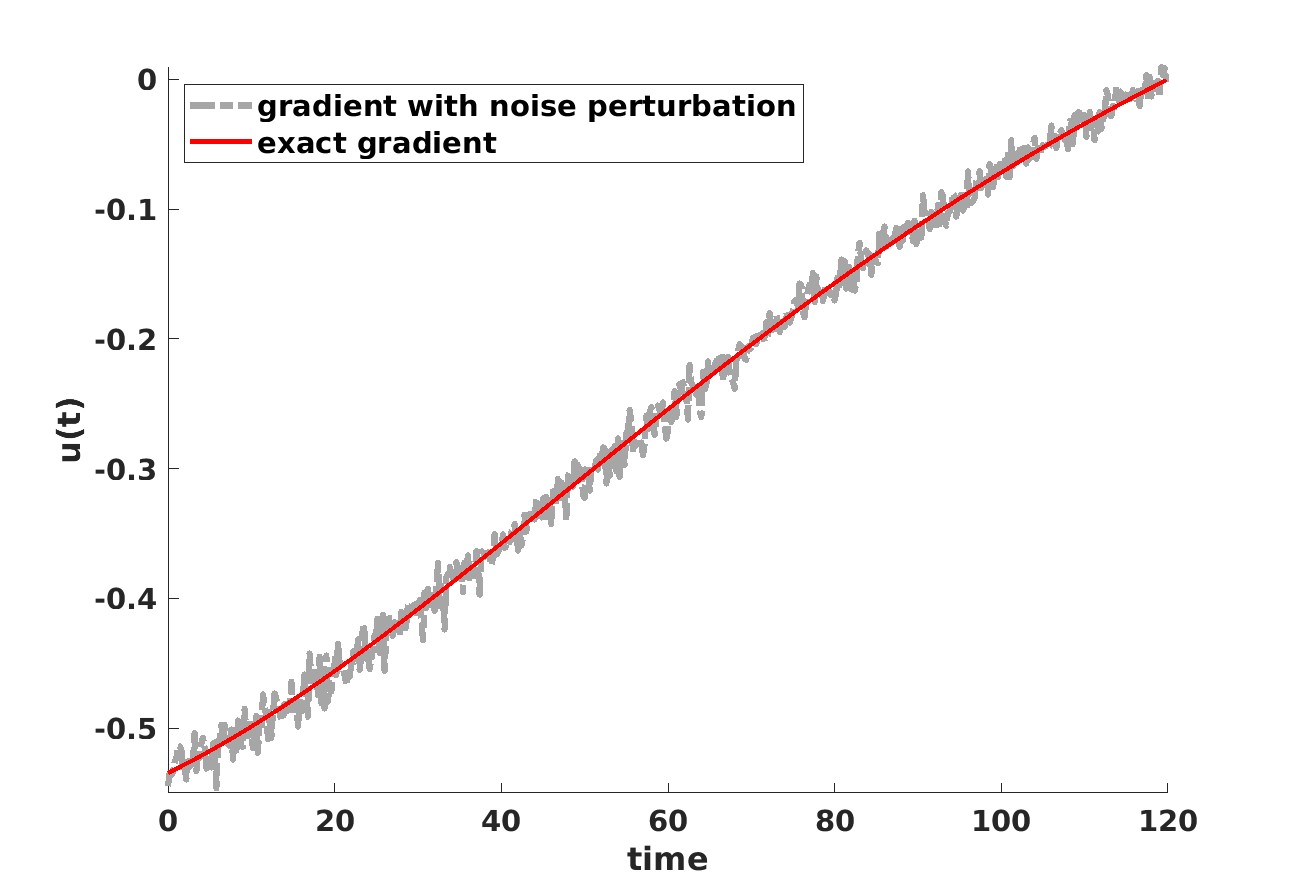}}
    \caption{The control function after 2000 iterations.}
    \label{fig:ctrl}
    \end{center}
    \vskip -0.2in
\end{figure}

\section{Acknowledgements}
We thank the anonymous reviewers for the valuable comments. We are also grateful to Wenhao He for pointing out a miscalculation of the gate complexity in an earlier version of this paper. XL's research is supported by the National Science Foundation Grants DMS-2111221. CW acknowledges support from National Science Foundation grant CCF-2238766 (CAREER). Both XL and CW were supported by a seed grant from the Institute of Computational and Data Science (ICDS).

\newpage
\bibliographystyle{plain}
\bibliography{qctrl,ctrl,qcomp}

\begin{thebibliography}{10}

\bibitem{albertini_lie_2002}
Francesca Albertini and Domenico D'Alessandro.
\newblock The {Lie} algebra structure and controllability of spin systems.
\newblock {\em Linear Algebra and its Applications}, 350(1-3):213--235, July
  2002.

\bibitem{banchi2021measuring}
Leonardo Banchi and Gavin~E Crooks.
\newblock Measuring analytic gradients of general quantum evolution with the
  stochastic parameter shift rule.
\newblock {\em Quantum}, 5:386, 2021.

\bibitem{berry2017quantum}
Dominic~W Berry, Andrew~M Childs, Aaron Ostrander, and Guoming Wang.
\newblock Quantum algorithm for linear differential equations with
  exponentially improved dependence on precision.
\newblock {\em Communications in Mathematical Physics}, 356(3):1057--1081,
  2017.

\bibitem{berry2020time}
Dominic~W Berry, Andrew~M Childs, Yuan Su, Xin Wang, and Nathan Wiebe.
\newblock Time-dependent {Hamiltonian} simulation with $ l^{1}$-norm scaling.
\newblock {\em Quantum}, 4:254, 2020.

\bibitem{brif_control_2010}
Constantin Brif, Raj Chakrabarti, and Herschel Rabitz.
\newblock Control of quantum phenomena: past, present and future.
\newblock {\em New Journal of Physics}, 12(7):075008, July 2010.

\bibitem{castro_controlling_2012}
A.~Castro, J.~Werschnik, and E.~K.~U. Gross.
\newblock Controlling the {Dynamics} of {Many}-{Electron} {Systems} from
  {First} {Principles}: {A} {Combination} of {Optimal} {Control} and
  {Time}-{Dependent} {Density}-{Functional} {Theory}.
\newblock {\em Physical Review Letters}, 109(15), October 2012.

\bibitem{castro2004propagators}
Alberto Castro, Miguel~AL Marques, and Angel Rubio.
\newblock Propagators for the time-dependent kohn--sham equations.
\newblock {\em The Journal of chemical physics}, 121(8):3425--3433, 2004.

\bibitem{cerezo2021variational}
Marco Cerezo, Andrew Arrasmith, Ryan Babbush, Simon~C Benjamin, Suguru Endo,
  Keisuke Fujii, Jarrod~R McClean, Kosuke Mitarai, Xiao Yuan, Lukasz Cincio,
  et~al.
\newblock Variational quantum algorithms.
\newblock {\em Nature Reviews Physics}, 3(9):625--644, 2021.

\bibitem{CGJ19}
Shantanav Chakraborty, Andr{\'a}s Gily{\'e}n, and Stacey Jeffery.
\newblock The power of block-encoded matrix powers: Improved regression
  techniques via faster {H}amiltonian simulation.
\newblock In {\em Proceedings of the 46th International Colloquium on Automata,
  Languages, and Programming (ICALP 2019)}, 2019.

\bibitem{childs2019quantum}
Andrew~M Childs and Jin-Peng Liu.
\newblock Quantum spectral methods for differential equations.
\newblock {\em arXiv preprint arXiv:1901.00961}, 2019.

\bibitem{choquette2021quantum}
Alexandre Choquette, Agustin Di~Paolo, Panagiotis~Kl Barkoutsos, David
  S{\'e}n{\'e}chal, Ivano Tavernelli, and Alexandre Blais.
\newblock Quantum-optimal-control-inspired ansatz for variational quantum
  algorithms.
\newblock {\em Physical Review Research}, 3(2):023092, 2021.

\bibitem{ciaramella_newton_2015}
G.~Ciaramella, A.~Borzì, G.~Dirr, and D.~Wachsmuth.
\newblock Newton {Methods} for the {Optimal} {Control} of {Closed} {Quantum}
  {Spin} {Systems}.
\newblock {\em SIAM Journal on Scientific Computing}, 37(1):A319--A346, January
  2015.

\bibitem{crooks2019gradients}
Gavin~E Crooks.
\newblock Gradients of parameterized quantum gates using the parameter-shift
  rule and gate decomposition.
\newblock {\em arXiv preprint arXiv:1905.13311}, 2019.

\bibitem{dalessandro_topological_2000}
Domenico D'Alessandro.
\newblock Topological properties of reachable sets and the control of quantum
  bits.
\newblock {\em Systems \& Control Letters}, 41(3):213--221, October 2000.

\bibitem{dalessandro_introduction_2008}
Domenico D'Alessandro.
\newblock {\em Introduction to quantum control and dynamics}.
\newblock Chapman \& {Hall}/{CRC} applied mathematics and nonlinear science
  series. Chapman \& Hall/CRC, Boca Raton, 2008.
\newblock OCLC: ocn123539646.

\bibitem{datta2005quantum}
Supriyo Datta.
\newblock {\em Quantum transport: {A}tom to transistor}.
\newblock Cambridge University Press, 2005.

\bibitem{deuflhard2002scientific}
Peter Deuflhard and Folkmar Bornemann.
\newblock {\em Scientific computing with ordinary differential equations},
  volume~42.
\newblock Springer Science \& Business Media, 2002.

\bibitem{Weinan2019}
W.~E, J.~Han, and Q.~Li.
\newblock A mean-field optimal control formulation of deep learning.
\newblock {\em Res. Math. Sci.}, 6:1--41, 2019.

\bibitem{eitan_optimal_2011}
Reuven Eitan, Michael Mundt, and David~J. Tannor.
\newblock Optimal control with accelerated convergence: {Combining} the
  {Krotov} and quasi-{Newton} methods.
\newblock {\em Physical Review A}, 83(5), May 2011.

\bibitem{farhi2014quantum}
Edward Farhi, Jeffrey Goldstone, and Sam Gutmann.
\newblock A quantum approximate optimization algorithm.
\newblock {\em arXiv preprint arXiv:1411.4028}, 2014.

\bibitem{ge2021optimization}
Xiaozhen Ge, Rebing Wu, and Herschel Rabitz.
\newblock Optimization landscape of quantum control systems.
\newblock {\em Complex System Modeling and Simulation}, 1(2):77--90, 2021.

\bibitem{geppert2004laser}
D~Geppert, L~Seyfarth, and R~de~Vivie-Riedle.
\newblock Laser control schemes for molecular switches.
\newblock {\em Appl. Phys. B}, 79(8):987--992, 2004.

\bibitem{gilyen2019optimizing}
Andr{\'a}s Gily{\'e}n, Srinivasan Arunachalam, and Nathan Wiebe.
\newblock Optimizing quantum optimization algorithms via faster quantum
  gradient computation.
\newblock In {\em Proceedings of the Thirtieth Annual ACM-SIAM Symposium on
  Discrete Algorithms}, pages 1425--1444. SIAM, 2019.

\bibitem{GSLW19}
Andr{\'a}s Gily{\'e}n, Yuan Su, Guang~Hao Low, and Nathan Wiebe.
\newblock Quantum singular value transformation and beyond: exponential
  improvements for quantum matrix arithmetics.
\newblock In {\em Proceedings of the 51st Annual ACM SIGACT Symposium on Theory
  of Computing}, pages 193--204. ACM, 2019.

\bibitem{glaser_training_2015}
Steffen~J. Glaser, Ugo Boscain, Tommaso Calarco, Christiane~P. Koch, Walter
  Köckenberger, Ronnie Kosloff, Ilya Kuprov, Burkhard Luy, Sophie Schirmer,
  Thomas Schulte-Herbrüggen, Dominique Sugny, and Frank~K. Wilhelm.
\newblock Training {Schrödinger}’s cat: quantum optimal control: {Strategic}
  report on current status, visions and goals for research in {Europe}.
\newblock {\em The European Physical Journal D}, 69(12), December 2015.

\bibitem{he2016deep}
K.~He, X.~Zhang, S.~Ren, and J.~Sun.
\newblock Deep residual learning for image recognition.
\newblock In {\em Proc. IEEE Comput. Soc. Conf. Comput. Vis. Pattern
  Recognit.}, pages 770--778, 2016.

\bibitem{hellgren_optimal_2013}
Maria Hellgren, Esa Räsänen, and E.~K.~U. Gross.
\newblock Optimal control of strong-field ionization with time-dependent
  density-functional theory.
\newblock {\em Physical Review A}, 88(1), July 2013.

\bibitem{jin2019short}
Chi Jin, Praneeth Netrapalli, Rong Ge, Sham~M Kakade, and Michael~I Jordan.
\newblock A short note on concentration inequalities for random vectors with
  subgaussian norm.
\newblock {\em arXiv preprint arXiv:1902.03736}, 2019.

\bibitem{jin2021nonconvex}
Chi Jin, Praneeth Netrapalli, Rong Ge, Sham~M Kakade, and Michael~I Jordan.
\newblock On nonconvex optimization for machine learning: Gradients,
  stochasticity, and saddle points.
\newblock {\em Journal of the ACM (JACM)}, 68(2):1--29, 2021.

\bibitem{jordan2005fast}
Stephen~P Jordan.
\newblock Fast quantum algorithm for numerical gradient estimation.
\newblock {\em Physical Review Letters}, 95(5):050501, 2005.

\bibitem{kieferova2019simulating}
M{\'a}ria Kieferov{\'a}, Artur Scherer, and Dominic~W Berry.
\newblock Simulating the dynamics of time-dependent hamiltonians with a
  truncated dyson series.
\newblock {\em Physical Review A}, 99(4):042314, 2019.

\bibitem{kincaid2009numerical}
David Kincaid, David~Ronald Kincaid, and Elliott~Ward Cheney.
\newblock {\em Numerical analysis: mathematics of scientific computing},
  volume~2.
\newblock American Mathematical Soc., 2009.

\bibitem{kyriienko2021generalized}
Oleksandr Kyriienko and Vincent~E Elfving.
\newblock Generalized quantum circuit differentiation rules.
\newblock {\em Physical Review A}, 104(5):052417, 2021.

\bibitem{larocca2022diagnosing}
Martin Larocca, Piotr Czarnik, Kunal Sharma, Gopikrishnan Muraleedharan,
  Patrick~J Coles, and Marco Cerezo.
\newblock Diagnosing barren plateaus with tools from quantum optimal control.
\newblock {\em Quantum}, 6:824, 2022.

\bibitem{larocca2021theory}
Martin Larocca, Nathan Ju, Diego Garc{\'\i}a-Mart{\'\i}n, Patrick~J Coles, and
  Marco Cerezo.
\newblock Theory of overparametrization in quantum neural networks.
\newblock {\em arXiv preprint arXiv:2109.11676}, 2021.

\bibitem{LPQ+22}
Jiaqi Leng, Yuxiang Peng, Yi-Ling Qiao, Ming Lin, and Xiaodi Wu.
\newblock Differentiable analog quantum computing for optimization and control.
\newblock In {\em Advances in Neural Information Processing Systems}, 2022.

\bibitem{li2017hybrid}
Jun Li, Xiaodong Yang, Xinhua Peng, and Chang-Pu Sun.
\newblock Hybrid quantum-classical approach to quantum optimal control.
\newblock {\em Physical review letters}, 118(15):150503, 2017.

\bibitem{LC19}
Guang~Hao Low and Isaac~L Chuang.
\newblock Hamiltonian simulation by qubitization.
\newblock {\em Quantum}, 3:163, 2019.

\bibitem{low2018hamiltonian}
Guang~Hao Low and Nathan Wiebe.
\newblock Hamiltonian simulation in the interaction picture.
\newblock {\em arXiv preprint arXiv:1805.00675}, 2018.

\bibitem{machnes2018tunable}
Shai Machnes, Elie Ass{\'e}mat, David Tannor, and Frank~K Wilhelm.
\newblock Tunable, flexible, and efficient optimization of control pulses for
  practical qubits.
\newblock {\em Physical review letters}, 120(15):150401, 2018.

\bibitem{maday_new_2003}
Yvon Maday and Gabriel Turinici.
\newblock New formulations of monotonically convergent quantum control
  algorithms.
\newblock {\em The Journal of Chemical Physics}, 118(18):8191--8196, May 2003.

\bibitem{mccreery2013critical}
Richard~L McCreery, Haijun Yan, and Adam~Johan Bergren.
\newblock A critical perspective on molecular electronic junctions: there is
  plenty of room in the middle.
\newblock {\em Phys. Chem. Chem. Phys.}, 15(4):1065--1081, 2013.

\bibitem{nesterov2018lectures}
Yurii Nesterov.
\newblock {\em Lectures on convex optimization}, volume 137.
\newblock Springer, 2018.

\bibitem{NC02}
Michael~A Nielsen and Isaac Chuang.
\newblock {\em Quantum Computation and Quantum Information}.
\newblock Cambridge University Press, 2010.

\bibitem{song2022optimizing}
Yao Song, Junning Li, Yong-Ju Hai, Qihao Guo, and Xiu-Hao Deng.
\newblock Optimizing quantum control pulses with complex constraints and few
  variables through autodifferentiation.
\newblock {\em Physical Review A}, 105(1):012616, 2022.

\bibitem{werschnik_quantum_2007}
J.~Werschnik and E.~K.~U. Gross.
\newblock Quantum {Optimal} {Control} {Theory}.
\newblock {\em arXiv:0707.1883 [quant-ph]}, July 2007.
\newblock arXiv: 0707.1883.

\bibitem{wierichs2022general}
David Wierichs, Josh Izaac, Cody Wang, and Cedric Yen-Yu Lin.
\newblock General parameter-shift rules for quantum gradients.
\newblock {\em Quantum}, 6:677, 2022.

\bibitem{yang2017optimizing}
Zhi-Cheng Yang, Armin Rahmani, Alireza Shabani, Hartmut Neven, and Claudio
  Chamon.
\newblock Optimizing variational quantum algorithms using pontryagin’s
  minimum principle.
\newblock {\em Physical Review X}, 7(2):021027, 2017.

\bibitem{zhou2018sgd}
Yi~Zhou, Junjie Yang, Huishuai Zhang, Yingbin Liang, and Vahid Tarokh.
\newblock Sgd converges to global minimum in deep learning via star-convex
  path.
\newblock In {\em International Conference on Learning Representations}, 2018.

\bibitem{zhu1998rapid}
Wusheng Zhu and Herschel Rabitz.
\newblock A rapid monotonically convergent iteration algorithm for quantum
  optimal control over the expectation value of a positive definite operator.
\newblock {\em J.Chem. Phys.}, 109(2):385--391, 1998.

\bibitem{zhu_rapid_1998}
Wusheng Zhu and Herschel Rabitz.
\newblock A rapid monotonically convergent iteration algorithm for quantum
  optimal control over the expectation value of a positive definite operator.
\newblock {\em The Journal of Chemical Physics}, 109(2):385--391, July 1998.

\end{thebibliography}

\newpage
\appendix
\onecolumn
\section{The proof of \cref{lemma:partial-alpha} }\label{sec: A1}

\begin{proof}
  To understand how the objective function $J$ depends on the control variables, we first explicitly write the piecewise-linear interpolating function using shape functions $b_j(t)$,
\[u_I(t) = \sum_j u_j b_j(t), \]
where $b_j(t)$'s are standard hat functions. We first examine the derivatives of $\ket{\psi(T)}$ with respect to the control variables.
Following the TDSE in \cref{eq: tdse}, we can write the derivative of $\psi(t)$ with respect to $u_j$, $\phi_1(t)\coloneqq \partial_{u_j} \psi(t),$ as follows, 
\begin{equation}
     i\partial_t \phi_1 = H(t) \phi_1  - \mu b_j(t) \psi(t).
\end{equation}
Known as a variational equation \cite{deuflhard2002scientific}, such an equation characterizes the dependence of the solution on the parameters in an evolution equation. By using variation-of-constant formula, we find
\[
\phi_1(t)= i \int_0^t U(t,t') \mu \psi(t')  b_j(t') \,\dd t', 
\] 
where $U(t,t')$ is the unitary operator generated by $H(t).$
A straightforward bound thus  follows,
\[
\norm{\phi(t)} \leq \norm{\mu} \delta, \; \forall  0\leq t \leq T. 
\]
Here we used the fact that 
\[ \int_0^T \abs{b_j(t)}\, \dd t \leq  \delta.  \]

We can continue with this calculation by letting $\phi_2(t)\coloneqq \partial_{u_j} \phi_1(t)$ for a second order derivative:
\[i\partial_t \phi_2 = H(t) \phi_2  - 2 \mu b_j(t) \phi_1(t).   \]
Therefore, $\norm{\phi_2(t)} \leq 2 (\delta \norm{\mu})^2.$ Mixed derivatives can be treated similarly and they follow a similar bound. 
By defining $\phi_3(t)\coloneqq \partial_{u_j} \phi_2(t)$, we obtain,
\[i\partial_t \phi_3 = H_0 \phi_3  - 3 \mu b_j(t) \phi_2(t).   \]
By repeating this argument for higher-order derivatives, we arrive at the inequality. 
\begin{equation}
  \norm{ \frac{\partial^{\bm{\alpha}}}{\partial u_{\alpha_1} u_{\alpha_2} \cdots u_{\alpha_k} } \psi(T)  } \leq  k! \left(\delta  \|\mu\| \right)^k.
\end{equation}
Consequently, the  function $\widetilde{J}_1$ has derivative bounds,
\begin{equation}
      \begin{aligned}
            \norm{ \frac{\partial^{\bm{\alpha}} \widetilde{J}_1}{\partial u_{\alpha_1} u_{\alpha_2} \cdots u_{\alpha_k} }   }  \leq  \sum_{j=0}^k  \binom{k}{j} (k-j)! j! \left(\delta  \|\mu\| \right)^k  
              \leq (k+1)! \left(\delta  \|\mu\| \right)^k.
      \end{aligned}
\end{equation}

\end{proof}

\section{Proof of \cref{lemma:grad-est-J}}
\label{sec:proof-grad-est-J}
\begin{proof}
  First observe that $\widetilde{J}_1(\bm{u})$ depends on $\bm{u}$ smoothly. By \cref{lemma:partial-alpha}, we have
  \begin{align}
    B = \abs{\nabla_{\bm r}^{2m+1} J} = \mathcal{O}(  (2m+1)! \delta^{2m+1}).
  \end{align}
  Therefore, according to \cref{lemma:jordan}, the query complexity to achieve the error bound in the spectral norm is
  \begin{align}
    \widetilde{\CO}\left(\frac{N}{\epsilon} \frac{B^{\frac{1}{2m} } N^{\frac{1}{4m} }  }{\epsilon^{\frac{1}{2m} }}\right).
  \end{align}
We notice that by Stirling's approximation,
\begin{align}
 (2m+1)!  \approx \left( \frac{2m+1}{e} \right)^{2m+1}.
\end{align}
Therefore, it suffices to choose
\begin{align}
  m = \CO\left(\frac{\log(T^{1/2}/\epsilon^{3/4})}{\log\log(T^{1/2}/\epsilon^{3/4})}\right),
\end{align}
so that both $m$ and $(T^{1/2}/\epsilon^{3/4})^{1/(2m)}$ are poly-logarithmic factors. As a result, we have the claimed query complexity.
\end{proof}

\section{The function ascent property \cref{Juk-Ju0} }\label{sec: A2}
 We start by noticing that,
\[
\begin{aligned}
     J(\bd{u}^{(k+1)}) & \geq J( \bd{u}^{(k)} ) + \inner{\nabla J(\bd{u}^{(k)})}{\bd{u}^{(k+1)}-\bd{u}^{(k)}}  - \frac{L}{2} \norm{\bd{u}^{(k+1)}-\bd{u}^{(k)}}^2 \\
      & \geq J( \bd{u}^{(k)} ) + 
      \eta \norm{\nabla J(\bd{u}^{(k)})}^2 + \eta
      \inner{\nabla J(\bd{u}^{(k)})}{\bd{\zeta}_{k}+\bd{e}_k}  - \frac{\eta}{2} \norm{\nabla J(\bd{u}^{(k)}) +\bd{e}_k  +  \bd{\zeta}_k }^2 \\
      & \geq  J( \bd{u}^{(k)} )  + \frac{\eta}{2} \norm{ \nabla J(\bd{u}^{(k)})}^2 
       - \eta \norm{\bd{\zeta_k}}^2  - \eta\norm{\bd{e}_{k}}^2. 
\end{aligned}
\]

The first line was obtained from the Lipschitz condition on the gradient.  \cref{sgd} is then used to arrive at the second line.

By a telescoping sum, we arrive at the desired function ascent property,
\begin{equation}\label{descent}
\begin{aligned}
    J( \bd{u}^{(k)} ) \geq& J( \bd{u}^{(0)} ) +\frac{\eta}{2} \sum_{j=0}^{k-1}   \norm{ \nabla J(\bd{u}^{(j)})}^2 
    \\ 
    & \quad - \eta  \sum_{j=0}^{k-1}  \norm{\bd{\zeta_j}}^2 -  \eta  \sum_{j=0}^{k-1}  \norm{\bd{e_j}}^2.
\end{aligned}
\end{equation}

The noise term $\sum_{j=0}^{k-1}  \norm{\bd{\zeta_j}}^2$ has been treated probabilistically in \cite{jin2021nonconvex} using a concentration inequality \cite{jin2019short}, which provides a bound of $\frac{1}{8}k\epsilon^2$ with high probability.  Meanwhile, we set a failure probability $\rho=\nu/k$ in the gradient estimation such that,
\[\mathbb{P}\big(\norm{\bd{e}_j}^2 < \epsilon^2/8\big) \geq (1 - \rho),
\]
implying that
\begin{equation}
    \mathbb{P}\big(\norm{\bd{e}_j}^2 < \epsilon^2/8, \;\forall 1\leq j \leq k  \big) \geq (1 - \rho)^k \geq 1 - \nu.
\end{equation}

\end{document}